\documentclass[11pt]{amsart}
\usepackage{amsfonts}
\usepackage{amssymb}
\usepackage[dvips]{graphics}
\usepackage{epsfig}
\usepackage{euscript}
\usepackage{color}

\usepackage{fancyvrb}

\usepackage[pdftex]{hyperref}

 \newtheorem{thm}{Theorem}[section]
 \newtheorem{cor}[thm]{Corollary}
 \newtheorem{lemma}[thm]{Lemma}
 \newtheorem{prop}[thm]{Proposition}
 \theoremstyle{definition}
 
 \theoremstyle{remark}

 \numberwithin{equation}{section}

\newcommand{\caA}{{\mathcal A}}
\newcommand{\caB}{{\mathcal B}}
\newcommand{\caC}{{\mathcal C}}

\newcommand{\caG}{{\mathcal G}}
\newcommand{\caH}{{\mathcal H}}

\newcommand{\caK}{{\mathcal K}}
\newcommand{\caL}{{\mathcal L}}
\newcommand{\caM}{{\mathcal M}}

\newcommand{\caS}{{\mathcal S}}

\newcommand{\caU}{{\mathcal U}}
\newcommand{\caV}{{\mathcal V}}

\newcommand{\bbC}{{\mathbb C}}

\newcommand{\bbI}{{\mathbb I}}

\newcommand{\bbN}{{\mathbb N}}

\newcommand{\bbR}{{\mathbb R}}
\newcommand{\bbS}{{\mathbb S}}

\newcommand{\bbZ}{{\mathbb Z}}

\newcommand{\iu}{\mathrm{i}}

\newcommand{\str}{^{*}}

\newcommand{\Tr}{\mathrm{Tr}}

\newcommand{\be}{\begin{equation}}
\newcommand{\ee}{\end{equation}}
\newcommand{\bea}{\begin{eqnarray}}
\newcommand{\eea}{\end{eqnarray}}
\newcommand{\beann}{\begin{eqnarray*}}
\newcommand{\eeann}{\end{eqnarray*}}

\begin{document}

\renewcommand{\thefootnote}{\fnsymbol{footnote}}
\title[Topological order and anyons]{Local disorder, topological ground state degeneracy and entanglement entropy, and discrete anyons}
\author{Sven Bachmann}
\address{Mathematisches Institut der Universit{\"a}t M{\"u}nchen \\ 80333 M{\"u}nchen \\ Germany}
\email{sven.bachmann@math.lmu.de}

\date{\today }

\begin{abstract}
In this comprehensive study of Kitaev's abelian models defined on a graph embedded on a closed orientable surface, we provide complete proofs of the topological ground state degeneracy, the absence of local order parameters, compute the entanglement entropy exactly and characterise the elementary anyonic excitations. The homology and cohomolgy groups of the cell complex play a central role and allow for a rigorous understanding of the relations between the above characterisations of topological order.
\end{abstract}

\maketitle

\section{Introduction}\label{sec:intro}

Despite being heavily studied both in concrete models and in more abstract settings, the very notion of topological order is still lacking a complete mathematical treatment. A comprehensive picture of the relationships between various characterisations of topological order is missing. Among them, one should mention (i) local disorder, namely the absence of local order parameters and the emergence of superselection sectors in the infinite volume limit, (ii) a ground state degeneracy that depends on the topology of the surface underlying the model, (iii) the presence of a topological term in the entanglement entropy, and (iv) the fact that elementary quasi-particles are anyons. The solvable two-dimensional topologically ordered models~\cite{Kitaev:2003ul,Levin:2005fy,Kitaev:2006ik} exhibit all of these characteristics and they are all -- heuristically -- related, see also~\cite{Chen:2010gb} for a renomalisation picture. In this article, we clarify these relationships in a rigorous mathematical manner in the case of the abelian models of Kitaev~\cite{Kitaev:2003ul}. We aim at exhibiting the structures that may survive more generally, in particular that may be stable under the spectral flow~\cite{Hastings:2005cs,Bachmann:2011kw}, thereby providing invariants for the classification of topological phases.

We analyse in detail Kitaev's models, namely the spin-$\frac{1}{2}(d-1), d\in\bbN,$ models based on the abelian groups $\bbZ_d = \bbZ/ d\bbZ$. We consider only finite volumes, where the quantum spins lie on the edges of a graph embedded on a connected, closed orientable surface of genus $g$. For a related discussion of the infinite volume limit in the setting of algebraic quantum field theory, see e.g.~\cite{Naaijkens2015}, and~\cite{Cha:2016vl} for a detailed discussion of the set of ground states. The toric code model was also studied in~\cite{Freedman:2001dp} in the case of the projective plane $\bbR P^2$, a compact surface with boundary. Although most of the results presented here can be found in some form in the literature~\cite{Kitaev:2003ul, Bombin:2008es}, we aim in this article for mathematical completeness, simplicity (for example, there is no need to introduce ribbon operators) and concentrate in particular on the role played by cycles and cocycles in these models: the homology and cohomology groups are indeed the essential element relating all aspects of topological order.

We should mention that braid statistics and topological superselection sectors were already discussed in great generality in the context of algebraic quantum field theory, see e.g.~\cite{Fredenhagen:1989wr,Frohlich:1990aa}. Although the general picture is heuristically the same and although some aspects of this local quantum physics point of view can be rigorously translated for quantum spin systems in the infinite volume limit, the finite volume systems analysed here exhibit different and interesting aspects of topological order. 

After introducing the model in Section~\ref{sec:Models}, we first consider its ground state space in Section~\ref{Sec:GSS} and prove that it carries a joint representation of the first homology and cohomology groups, Theorem~\ref{thm:SBraid}. In Section~\ref{GSEE}, we compute the ground state entanglement entropy of any finite subset of spins, Theorem~\ref{thm:EE}, and give a simple explanation of both the area law and the topological entanglement entropy in terms of cycles. Finally, we turn to excited states in Section~\ref{sec:Excited states}. We identify elementary excitations which are string-localised eigenstates. In~Theorem~\ref{thm:ExcitedBraids}, we prove that the space of two particle-antiparticle pairs is unitarily related to the ground state space, and that their local braiding is equal to a combination of topological operators. This is our main result, as it exhibits the equivalence of the topological ground state degeneracy and the anyonic nature of the (distinguishable) particles. This had already been observed in~\cite{Kitaev:2003ul}, although only in the continuum limit where the braid group ought to arise.

Finally, we note that for non-interacting fermions having topological properties -- quantum Hall systems and topological insulators -- topological order is also manifest in the bulk-edge correspondence~\cite{Ryu:2002gj,Qi:2006jm, Graf:2013ju} and index theorems~\cite{Prodan:2016aa}. In the context of quantum spin systems, bulk-edge correspondence remains a mostly unexplored field in general, but \cite{Bachmann:2014fl,Bachmann:2015ia} has showed that even trivial bulk phases can have interesting edge excitations, while \cite{Beigi:2011kr,Kitaev:2012bp} have considered topologically ordered models with boundaries. Although a complete understanding of topological phases will indeed require the bulk-edge duality to be considered in details, we shall leave this aspect aside in the present article. 

\section{Kitaev's abelian models}\label{sec:Models}

Kitaev's Hamiltonians are naturally defined on an arbitrary cell decomposition of a two-dimensional surface. We present the geometric and topological setting and define the quantum spin system and the Hamiltonian.

\subsection{A quantum spin system}

We shall define the model on an arbitrary connected closed orientable surface. It is topologically completely characterised by its genus, being homeomorphic to the sphere with $g$ handles $T_g$, for some $g\in\bbN$ (i.e.~$T_0 = \bbS^2$ is the $2$-sphere, $T_1 = \bbS^1\times\bbS^1$ is the torus, and so on). We consider an arbitrary polygon decomposition $G$ of $T_g$, namely
\begin{equation*}
G = \caG_0\cup\caG_1\cup\caG_2,
\end{equation*}
where $\caG_0$ is a set of \emph{vertices}, $\caG_1$ a set of \emph{edges} and $\caG_2$ a set of \emph{faces}. Note that in this definition, the union is not disjoint: a face contains its edges, an edge its vertices, and we shall write, slightly abusing notation $e\in f$ or $e\ni v$ to indicate that the edge $e$ is part of the face $f$ or that $v$ is one of its vertices. The decomposition is such that any edge lies between exactly two different vertices, and is shared by exactly two different faces, as would be the case for a triangulation of $T_g$. The edges carry an arbitrary but fixed orientation, while the faces inherit the orientation of the surface. We shall generally refer to an element in $\caG_n$ as an \emph{$n$-cell}, $n=0,1,2$, and they are all homeomorphic to an $n$-disc $D_n:=\{x\in\bbR^n:\vert x \vert\leq 1\}$.  A decomposition $G\str$ of $T_g$ that is dual to $G$ can be obtained by exchanging the roles of vertices and faces, namely by placing a vertex $v\str\in\caG^0$ in each face of $\caG_2$, joining dual vertices by a dual edge $e\str\in\caG^1$ through each edge in $\caG_1$, thereby determining a set of dual faces $f\str\in\caG^2$ which is in one-to-one correspondence with the set of vertices $\caG_0$. The dual edges are oriented so that $(e,e\str,n)$ form a right-handed trihedron, where $n$ is the outward normal of the surface. 

This geometric structure being set, we can define a quantum spin system where the spins live on the edges of $G$, and with two types of interactions: one among the spins shared by a vertex and one among the spins around a face. For each edge $e\in\caG_1$ of $G$, let
\begin{equation*}
\caH_e := \bbC^d
\end{equation*}
be the state space of a quantum spin-$\frac{1}{2}(d-1)$, and the total Hilbert space of the model is
\begin{equation*}
\caH_G := \otimes_{e\in\caG_1}\caH_e.
\end{equation*}
Note that any spin can equivalently be associated to a dual edge $e\str$. The matrix algebra of linear operators on~$\caH_G$ shall be denoted $\caA_G$. If $\caK\subset \caH_G$ is a subspace, we denote $\caA(\caK):=\{O\in\caA_G: O\caK\subset\caK\}$. We denote $\{l_i\}_{i=0}^{d-1}$ an orthonormal basis of~$\caH_e$. For the purpose of embedding the abelian models in the broader class of the quantum double models, one may identify $\caH_e$ with the group ring $\bbC[\bbZ_d]$ of the cyclic group $\bbZ_d:=\bbZ/d\bbZ$ over the complex numbers. All additions below are $(mod\ d)$. We define two elementary operators on $\caH_e$ by their action on the basis, namely
\begin{equation}\label{def:XZ}
X l_i = l_{i+1},\qquad Z l_i = \omega^i l_i,\qquad (0\leq i\leq d-1),
\end{equation}
where $\omega = \exp(2\pi\iu /d)$. Note that $X$ and $Z$ are unitary matrices and that 
\begin{equation}\label{XZ}
ZX = \omega XZ.
\end{equation}
For any vertex $v\in \caG_0$, and any face $f\in\caG_2$, let
\begin{equation}\label{litte ab}
a_v:= \prod_{e\ni v} X_{e}^{\varepsilon(e)},\qquad b_f:= \prod_{e\in f} Z_{e}^{\varepsilon(e)}
\end{equation}
where $X_{e}^{\varepsilon(e)}$ is $X_e$, resp.~$X_e^{-1} = X_e\str$ if the edge is incoming, resp.~outgoing, from $v$, and $Z_{e}^{\varepsilon(e)}$ is $Z_e$, resp.~$Z_e\str$, if the edge corresponds to the positive, resp. negative, orientation of the face (as we shall see, this is in fact best written in terms of $1$-(co)cycles). Finally, let
\begin{equation}\label{vertex and face operators}
A_v := \frac{1}{d}\sum_{j=0}^{d-1}(a_v)^j,\qquad B_f := \frac{1}{d}\sum_{j=0}^{d-1} (b_f)^j,
\end{equation}
and the Hamiltonian of Kitaev's model is given by
\begin{equation*}
H_G:= -\sum_{v\in\caG_0}A_v - \sum_{f\in\caG_2}B_f.
\end{equation*}
%

\subsection{Homology and cohomology groups}

We briefly recall the topological notions associated with cell complexes that will play an important role in the analysis. The \emph{$n$-chain group} $\caC_n(G;\bbZ_d)$ is the abelian group generated by oriented $n$-cells, namely $c\in \caC_n(G;\bbZ_d)$ is expressed as $c = \sum_{\sigma\in \caG_n} c_\sigma \sigma$, where $c_\sigma\in\bbZ_d$. Note that $-\sigma$ is the $n$-cell $\sigma$ with the opposite orientation. 
We denote the boundary of an $n$-cell $\sigma$ by $\partial_n \sigma$ (analytically, if $\sigma$ is the image of the homeomorphism $\varphi_\sigma$ defined on the disc $D_n$, then $\partial_n\sigma = \varphi_\sigma(\partial D_n)$), and it is by construction an $(n-1)$-chain. In particular, if $e$ is an edge oriented from $v$ to $\tilde v$, then $\partial_1 e = \tilde v - v$. As noted earlier, faces inherit the orientation of the surface, see Figure~\ref{Fig:Bd}. With this, $\partial_n$ can be extended to a linear operator $\partial_n: \caC_n(G;\bbZ_d)\to \caC_{n-1}(G;\bbZ_d)$. We shall usually drop the index $n$, and note that $\partial^2 = 0$.
\begin{figure}
\includegraphics{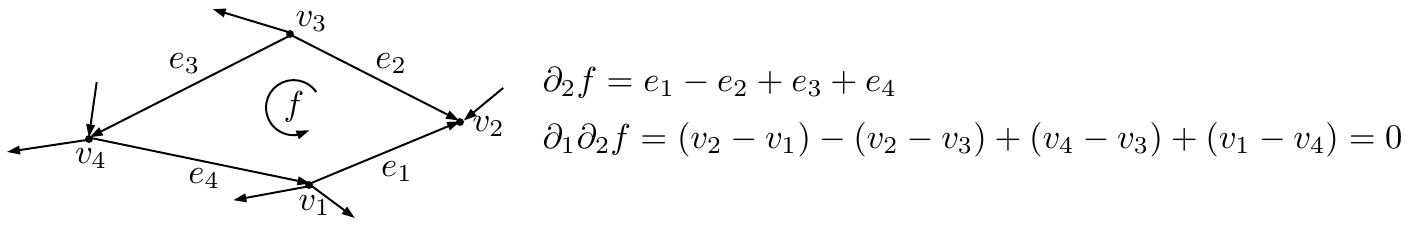}
\caption{An oriented $2$-cell $f$ of the polygon decomposition $G$, its boundary $\partial f$ which is a $1$-chain, and the fact that $\partial^2 f =0$ in this concrete case.}
\label{Fig:Bd}
\end{figure}

An \emph{$n$-cycle} is a $c\in \caC_n(G;\bbZ_d)$ such that $\partial c = 0$, and we denote the set thereof by $\caL_n(G;\bbZ_d)$ which is a subgroup of $\caC_n(G;\bbZ_d)$. An \emph{$n$-boundary} is a $b\in \caC_n(G;\bbZ_d)$ for which there exists a $c\in \caC_{n+1}(G;\bbZ_d)$ such that $b = \partial c$. It is also a subgroup $\caB_n(G;\bbZ_d)$ of $\caC_n(G;\bbZ_d)$. With this, the \emph{$n$th homology group} is
\begin{equation*}
H_n(G;\bbZ_d) := \caL_n(G;\bbZ_d) / \caB_n(G;\bbZ_d),
\end{equation*}
namely the set of equivalence classes of $n$-cycles, where two cycles are equivalent if they differ by a boundary. Finally, the \emph{$n$th Betti number} is the cardinality of the smallest generating set of $H_n(G;\bbZ_d)$. We have
\begin{equation*}
b_1(G;\bbZ_d) = 2g
\end{equation*}
for all $d\in\bbN$.

The same notions can be associated with the dual cell decomposition of $T_g$. Through $G\str$, one naturally defines the cochain groups $\caC^n(G;\bbZ_d)$, the cocycles and coboundaries and the \emph{$n$th cohomology group}
\begin{equation*}
H^n(G;\bbZ_d) := H_n(G\str;\bbZ_d),
\end{equation*}
namely the $n$th homology group of the dual complex. 

Finally, we note that $H_1(G;\bbZ_d)\cong \bigoplus_{i=1}^g\left(\bbZ_d\oplus \bbZ_d\right)$, so that any homology class can be represented as a pair $(\alpha,\beta)\in(\bbZ_d)^{g}\oplus(\bbZ_d)^{g}$, geometrically corresponding to $1$-cycles winding around the $i$th hole of the torus $(\alpha_i,\beta_i)$ times in the two possible directions. The same holds for cohomology classes. 

As we shall see, the structure of the quantum spin models depends intimately on the homology and cohomology groups. For more details on these algebraic topological aspects, we refer to~\cite{Nakahara:2003aa}.

\section{Topological ground state degeneracy}\label{Sec:GSS}
This section is dedicated to the explicit construction of the ground state space $\caS_G$ of $H_G$. It exhibits topological degeneracy, which is intimately related to the (co)homology groups.
\begin{lemma}\label{XZ_Basis}
(i) The matrices $X,Z$ defined in~(\ref{def:XZ}) generate the full matrix algebra $\caM_d(\bbC)$; \\ (ii) The linear span of $\{X^k Z^j:k,j\in\bbZ_d\}$ is $\caM_d(\bbC)$.
\end{lemma}
\begin{proof}
It suffices to prove $(ii)$, since $(ii) \Rightarrow (i)$. We first note that
\begin{equation}\label{Ts}
T_k:= \frac{1}{d}\sum_{j=0}^{d-1}(\omega^{-k} Z)^j
\end{equation}
are orthogonal projectors onto $l_k$, since
\begin{equation*}
T_k l_h= \frac{1}{d}\sum_{j=0}^{d-1}(\omega^{-k})^j(\omega^{h})^j l_h = \delta_{k,h}l_h,
\end{equation*}
and $T_k = T_k\str$. Furthermore, for $0\leq j\leq d-1$, let
\begin{equation}\label{Ls}
L_j:= X^j.
\end{equation}
Since $L_j l_i = l_{i+j}$, the matrices $\{L_{j-i}T_i\}_{i,j=0}^{d-1}$ form the canonical vector space basis of $\caM_d(\bbC)$:
\begin{equation*}
\langle l_k, L_{j-i} T_i l_h\rangle = \delta_{k,j}\delta_{i,h}.
\end{equation*}
This proves the claim since $L_{j-i}T_i$ is a linear combination of $X^{j-i}$ and $\{Z^k\}_{k=0}^{d-1}$.
\end{proof}
Since observables acting on different edges commute, the operators
\begin{equation*}
X_{\gamma\str} = \prod_{e\str\in\caG^1} X_{e\str}^{c_{e\str}},\qquad Z_\gamma := \prod_{e\in\caG_1} Z_{e}^{c_e},
\end{equation*}
are unequivocally defined for any $1$-cochain $\gamma\str = \sum_{e\str\in\caG^1}c_{e\str} e\str$ and $1$-chain $\gamma = \sum_{e\in\caG_1}c_e e$. In fact, they form a unitary group representation of the cochain group~$\caC^1(G;\bbZ_d)$, resp.~chain group~$\caC_1(G;\bbZ_d)$ on $\caH_G$ since
\begin{equation*}
X_{\gamma_1\str+ \gamma_2\str} = X_{\gamma_1\str}X_{\gamma_2\str},\qquad Z_{\gamma_1+ \gamma_2} = Z_{\gamma_1}Z_{\gamma_2}.
\end{equation*}
As noted earlier, the operators (\ref{litte ab}) are now naturally expressed as $b_f = Z_{\partial f}$ and $a_v = X_{\partial\str f\str}$, where $f\str$ is the face in $\caG^2$ dual to $v$. 
\begin{prop}\label{prop: all A span}
The set $\{X_{\gamma\str}Z_\gamma: \gamma\str\in\caC^1(G;\bbZ_d),\gamma\in \caC_1(G,\bbZ_d)\}$ linearly spans $\caA_G$.
\end{prop}
\begin{proof}
This is an immediate consequence of Lemma~\ref{XZ_Basis}, the fact that $X^d = 1 = Z^d$, and the commutativity of operators acting on different sites.
\end{proof}
We shall denote $|v|$ the degree of the vertex, and similarly $\vert f\vert$ the number of edges around a face, which is equal to the degree of the corresponding dual vertex.
\begin{lemma}\label{lem:AB Rep}
With $T_k,L_j,\,0\leq j,k\leq d-1$ defined in~(\ref{Ts},\ref{Ls}), and for any $v\in\caG_0,f\in\caG_2$, the following relations hold:
\begin{equation*}
B_f = \sum_{\substack{0\leq k_1,\ldots, k_{| f |}\leq d-1: \\
\sum_{e\in f} \varepsilon(e) k_e = 0
}}  \prod_{e\in f}T_{e,k_e},\qquad 
A_v = \frac{1}{d}\sum_{j=0}^{d-1}\prod_{e\ni v}L_{e,\varepsilon(e) j},
\end{equation*}
where $T_{e,k_e},L_{e,\varepsilon(e) j}\in\caA_e$. Moreover $B_f$ and $A_v$ are orthogonal projectors, and
\begin{equation*}
[B_f,A_v] = 0.
\end{equation*}
\end{lemma}
\begin{proof}
We first note that $Z = \sum_{j=0}^{d-1}\omega^j T_j$, and that $Z^k = \sum_{j=0}^{d-1}\omega^{kj} T_j$ since $T_j T_k = \delta_{j,k} T_k$.  Then,
\begin{equation*}
B_f = \frac{1}{d}\sum_{j=0}^{d-1}\sum_{0\leq k_1,\ldots, k_{| f |}\leq d-1}\omega^{j \sum_{e\in f} \varepsilon(e) k_e}\prod_{e\in f}T_{e,k_e},
\end{equation*}
which yields the claim by carrying out the sum over $j$. The self-adjointness of $B_f$ now follows from that of $T_k$ and the fact that the $T_{k_i}$'s in the product commute as they act on different vertices. Finally, $B_f$ is a projector since $\{T_k\}_{k=0}^{d-1}$ are mutually orthogonal projections.

The representation of $A_v$ follows from $(X\str)^j = X^{-j} = L_{-j}$. $A_v\str$ involves a product over $L_{e,-\varepsilon(e)j}$, and their sum over $j$ yields $A_v$ again. Furthermore,
\begin{equation*}
A_v^2 = \frac{1}{d^2}\sum_{j,j' = 0}^{d-1}\prod_{e\ni v}L_{e,\varepsilon(e)(j+j')}
\end{equation*}
and since the summation is modulo $d$, the sum $j+j'$ runs over $0$ to $d-1$ for any fixed $j$, so that $A_v^2 = A_v$.

Finally, if $v,f$ share no edge at all, then the commutation $[B_f,A_v]=0$ is trivial. In $G$, a face shares exactly two edges with any of its vertices, one of them, $e$, having the same orientation in both, while the other, $e'$, has the opposite orientation with respect to the vertex as with respect to the face. Hence, the a priori non-commuting factors in the product $b_f a_v$ are given by
\begin{equation*}
Z_e^{\varepsilon(e)}Z_{e'}^{\varepsilon(e')} X_e^{\varepsilon(e)}X_{e'}^{-\varepsilon(e')} = \omega \omega^{-1}X_e^{\varepsilon(e)}X_{e'}^{-\varepsilon(e')}Z_e^{\varepsilon(e)}Z_{e'}^{\varepsilon(e')} 
\end{equation*}
by~(\ref{XZ}), which yields the claim.
\end{proof}
In order to keep the notation as simple as possible, we shall write $T_{k_e}$ instead of $T_{e,k_e}$ in the following and similarly for the $L$ operators.

The Hamiltonian is thus a sum of commuting orthogonal projections, so that its spectrum is
\begin{equation*}
\mathrm{Spec}(H_G) \subset \left\{0,-1,-2,\ldots,-\left(\vert \caG_0\vert + \vert \caG_2\vert\right)\right\},
\end{equation*}
and its ground state space~$\caS_{G}$ is characterised by the following condition:
\begin{equation*}
\caS_{G}=\left\{\psi\in\caH_G: B_f\psi = \psi, A_v\psi = \psi,\,\forall f\in\caG_2,v\in\caG_0\right\}.
\end{equation*}
This is usually referred to as the \emph{frustration-free} property, namely the fact that any global ground state is also a ground state of all the local interaction terms. We also note that $\caS_{G}$ is not empty as will be proved in Theorem~\ref{thm:SBraid}, with an explicit basis later given in Proposition~\ref{ExplicitGS}.
\begin{lemma}\label{lem:little ab}
$\caS_{G}$ is equivalently characterised by
\begin{equation*}
\caS_{G}=\left\{\psi\in\caH_G: b_f\psi = \psi, a_v\psi = \psi,\,\forall f\in\caG_2,v\in\caG_0\right\}.
\end{equation*}
where $a_v,b_f$ were defined in~(\ref{litte ab}).
\end{lemma}
\begin{proof}
It suffices to prove that $B_f\psi = \psi$ is equivalent to $b_f\psi = \psi$, and similarly for $A_v$ operators. If $\psi$ is an eigenvector of $b_f$ for the eigenvalue $1$, then the definition~(\ref{vertex and face operators}) implies that $B_f\psi = d^{-1}\sum_{j=0}^{d-1}\psi = \psi$. Reciprocally, if $B_f\psi = \psi$, then the only non-zero coefficients $\psi_{k_1,\ldots, k_{| f |}}$ of $\psi$ in the tensor product basis $l_{k_1}\otimes\cdots\otimes l_{k_{| f |}}$ in the face $f$ are such that $\sum_{e\in f} \varepsilon(e) k_e = 0$ by Lemma~\ref{lem:AB Rep}. It follows that
\begin{equation*}
b_f\psi = \sum_{0\leq k_1,\ldots, k_{| f |}\leq d-1} \psi_{k_1,\ldots, k_{| f |}} \omega^{\sum_{e\in f} \varepsilon(e) k_e} l_{k_1}\otimes\cdots\otimes l_{k_{| f |}} = \psi.
\end{equation*}
The same holds for the $A_v$ operators by expanding $\psi$ in the tensor product eigenbasis of $X$.
\end{proof}

In the sequel, a (co)cycle shall refer to a $1$-(co)cycle if not otherwise specified.
\begin{thm}\label{thm:GS Alg}
$X_{\gamma\str}Z_\gamma\caS_G\subset\caS_G$ if and only if $\gamma$ is a cycle and $\gamma\str$ is a cocycle, and $\caA(\caS_G)$ is linearly spanned by $\left\{X_{\gamma\str}Z_\gamma: \gamma\str\in\caL^1(G;\bbZ_d),\gamma\in\caL_1(G;\bbZ_d)\right\}$.
\end{thm}
\begin{proof}
For $0\leq j\leq d-1$, let
\begin{equation*}
B_f(j) := \sum_{\substack{0\leq k_1,\ldots, k_{| f |}\leq d-1: \\
\sum_{e\in f} \varepsilon(e) k_e = j
}} \prod_{e\in f}T_{k_e},\qquad 
A_v(j) := \frac{1}{d}\sum_{k=0}^{d-1}\omega^{kj}\prod_{e\ni v}L_{\varepsilon(e) k}.
\end{equation*}
so that $B_f = B_f(0)$ and $A_v = A_v(0)$. First of all, $B_f(j)B_f(k) = \delta_{j,k} B_f(j)$ by the orthogonality of the projectors $T_k$. Secondly,
\begin{equation*}
A_v(j)A_v(k) = \frac{1}{d^2}\sum_{p,q=0}^{d-1}\omega^{jp + kq}\prod_{e\ni v}L_{\varepsilon(e)(p+q)} = A_v(j)\frac{1}{d}\sum_{p}\omega^{(k-j)p} = \delta_{j,k} A_v(j).
\end{equation*}
Furthermore, if $v$ and $\gamma$ share just one edge $e_0$, and it appears with coefficient $c_{e_0}\neq 0$ in $\gamma$, then
\begin{equation*}
A_v(j) Z_\gamma = Z_\gamma A_v(j + \varepsilon(e_0)c_{e_0})
\end{equation*}
and similarly $B_f(j) X_{\gamma\str} = X_{\gamma\str}B_f(j+ \varepsilon(e_0)c_{e_0})$. If $\gamma$ is not a cycle, namely $\gamma\in \caC_1(G;\bbZ_d)\setminus \caL_1(G;\bbZ_d)$, then it has a boundary, namely there is an edge $e_0\in\partial\gamma$ and a vertex $v\in e_0$ such that $v\in\gamma$ and $\{e'\ni v: e'\neq e_0\}\cap\gamma = \emptyset$. Then,
for any $\psi,\phi\in\caS_G$,
\begin{equation*}
\langle\phi, Z_{\gamma}\psi\rangle = \langle\phi, A_v Z_\gamma A_v \psi\rangle = \langle\phi, A_v(0)A_v(\pm c_{e_0})  Z_\gamma \psi\rangle = 0
\end{equation*}
and similarly for the action of $X_{\gamma\str}$ with $B_f$ for a suitably chosen face $f$. Hence, if either $\gamma$ or $\gamma\str$ is not a cycle, then $X_{\gamma\str}Z_\gamma\caS_G \perp\caS_G$.

If on the other hand $\gamma\in\caL_1(G;\bbZ_d),\gamma\str\in\caL^1(G;\bbZ_d)$, then $[A_v,X_{\gamma\str}Z_{\gamma}] = 0 = [B_f,X_{\gamma\str}Z_{\gamma}]$ for all $v\in\caG^{0},f\in\caG^{2}$. This follows from the fact that any $A_v$ commutes with any $X_{\gamma\str}$, and that the vertex $v$ must share an even number of edges with $\gamma$ with opposite orientations. Precisely, $\partial \gamma = 0$ implies that if $\gamma = \sum c_e e$, then for any $v\in\caG_0$, $\sum_{e\ni v} \varepsilon(e) c_e = 0$. Since, moreover,
\begin{equation}\label{ZaCommutation}
\prod_{e\ni v}Z_e^{c_e}\prod_{e\ni v}X_e^{\varepsilon(e)} = \omega^{\sum_{e\ni v}c_e\varepsilon(e)}\prod_{e\ni v}X_e^{\varepsilon(e)} \prod_{e\ni v}Z_e^{c_e},
\end{equation}
it follows that $[Z_\gamma,a_v] = 0$ and hence $[Z_\gamma,A_v] = 0$. Similarly, any $B_f$ commutes with any $Z_\gamma$, and the face $f$ shares an even number of edges with $\gamma\str$ with opposite orientations, so that $[B_f, X_{\gamma\str}]=0$. Hence if $\psi\in\caS_G$, then $\phi = X_{\gamma\str}Z_{\gamma}\psi$ is an eigenvector for all $A_v,B_f$ for the eigenvalue $1$, so that $\phi\in\caS_G$.

The last claim follows from the above and Proposition~\ref{prop: all A span}.
\end{proof}

Before we pursue, we note that Theorem~\ref{thm:GS Alg} can in particular be seen as a mathematical statement of Property~(i) of topological order mentioned in the introduction, namely the absence of any local order parameter. We shall call an observable $O\in\caA_G$ \emph{local} if all (co)cycles within $\mathrm{supp}(O)$ are (co)boundaries, namely $\mathrm{supp}(O)$ does not wind around any hole of the torus.
\begin{cor}
Let $O$ be a local observable. Then there exists $c(O)\in\bbC$ such that
\begin{equation}\label{LTQO}
P_{\caS_G} O P_{\caS_G} = C(O)P_{\caS_G}.
\end{equation}
In particular, $\langle \Psi, \Psi\rangle^{-1} \langle \Psi, O\Psi\rangle$ is independent of $\Psi\in\caS_G$.
\end{cor}
\begin{proof}
Since $O\in\caA_G$, $O = \sum_{\gamma\in \caC_1(G;\bbZ_d),\gamma\str\in \caC^1(G;\bbZ_d)} c_{\gamma,\gamma\str}X_{\gamma\str}Z_\gamma$ by Proposition~\ref{prop: all A span}. Theorem~\ref{thm:GS Alg} further implies that $P_{\caS_G} OP_{\caS_G} = \sum_{\gamma\in \caL_1(G;\bbZ_d),\gamma\str\in \caL^1(G;\bbZ_d)} c_{\gamma,\gamma\str}P_{\caS_G}X_{\gamma\str}Z_\gamma P_{\caS_G}$. Since $O$ is local, the only non-zero coefficients $c_{\gamma,\gamma\str}$ correspond to (co)boundaries, for which $P_{\caS_G}X_{\gamma\str}Z_\gamma P_{\caS_G} = P_{\caS_G}$. This yields~(\ref{LTQO}) with $C(O) = \sum_{\gamma\in \caL_1(G;\bbZ_d),\gamma\str\in \caL^1(G;\bbZ_d)} c_{\gamma,\gamma\str}$.
\end{proof}

\begin{lemma}\label{lem: GS Classes}
Let $\gamma,{\gamma'}\in\caL_1(G;\bbZ_d)$. If $[\gamma] = [\gamma']$, then
\begin{equation*}
Z_\gamma\upharpoonright{\caS_G} = Z_{\gamma'}\upharpoonright{\caS_G}.
\end{equation*}
A similar statement holds for cocycles and the $X$ operators.
\end{lemma}
\begin{proof}
Let $\beta$ be a 1-boundary, namely $\beta=\partial \left(\sum_{i=1}^N c_if_i\right), c_i\in\bbZ_d, f_i\in\caG_2$. We claim that 
\begin{equation*}
Z_\beta = \prod_{i=1}^N (b_{f_i})^{c_i}.
\end{equation*}
This follows by induction from the case of two faces sharing one edge where equality holds because the shared edge appears with opposite orientations with respect to the two faces, and the fact that $Z_e$ is unitary. By Lemma~\ref{lem:little ab},
\begin{equation}\label{Z Boundary}
Z_{\beta}\upharpoonright{\caS_G} =  \prod_{i=1}^N (b_{f_i})^{c_i}\upharpoonright{\caS_G} = \bbI_{\caS_G}.
\end{equation}
Now, if $[\gamma] = [\gamma']$, then there is a boundary  such that $\gamma' = \gamma + \beta$, and the claim follows from $Z_{\gamma'} = Z_\gamma Z_\beta$ and~(\ref{Z Boundary}). 
\end{proof}
In other words, the operators acting on the ground state space, which are associated with (co)cycles, are completely classified by their (co)homology class. As the proof above indicates, the class of the neutral element, namely $1$-boundaries, corresponds to the identity. The generators of the first homology group shall be described as follows: let $\lambda_i,\tau_j\in \caC_1(G;\bbZ_d)$ be elementary cycles that wind once around the $i$th hole of the genus~$g$ surface in the two possible directions. We shall write $[\lambda_i] = v_i\oplus 0$ and $[\tau_j] = 0\oplus v_j$, where $\{v_i\}_{i=1}^{g}$ denotes the canonical basis of $\bbR^g$. 

These observations naturally lead to the following description of the ground state space of Kitaev's abelian models.
\begin{thm}\label{thm:SBraid}
(i) The ground state space $\caS_G$ carries the following unitary representations: \\
\begin{minipage}{0.45\textwidth}
\begin{align*}
\pi:H_1(G;\bbZ_d)&\longrightarrow\caU(\caS_G) \\
[\gamma]&\longmapsto \pi([\gamma]) = Z_\gamma\upharpoonright\caS_G
\end{align*}
\end{minipage}
, resp. 
\begin{minipage}{0.45\textwidth}
\begin{align*}
\pi\str:H^1(G;\bbZ_d)&\longrightarrow\caU(\caS_G) \\
[\gamma\str]&\longmapsto \pi\str([\gamma\str]) = X_{\gamma\str}\upharpoonright\caS_G
\end{align*}
\end{minipage}.
\\
(ii) The algebra $\caA(\caS_G)$ is generated by
\begin{equation*}
\left\{X_{v_i,0}, X_{0,v_i}, Z_{v_i,0}, Z_{0,v_i}: 1\leq i\leq g\right\}
\end{equation*}
satisfying the relations
\begin{equation*}
(X_{\alpha,\beta})^d = \mathbb{I} = (Z_{\alpha,\beta})^d
\end{equation*}
and
\begin{equation*}
[[Z_{v_i,0},X_{0,v_i}]] = \omega, \qquad [[Z_{0,v_i},X_{v_i,0}]] = \omega,
\end{equation*}
where $[[A,B]] := ABA^{-1}B^{-1}$, and all other commutators vanish. \\
(iii) $\mathrm{dim}(\caS_G) = d^{b_1(G;\bbZ_d)} = d^{2g}$. \\
(iv) Any $\Psi\in\caS_G$ is cyclic in $\caS_G$ for $\left\{X_{\tau\str}Z_{\tau}: [\tau\str]\in H^1(G;\bbZ_d),[\tau]\in H_1(G;\bbZ_d)\right\}$, namely
\begin{equation*}
\caS_G= \mathrm{span}\left\{X_{\tau\str}Z_{\tau}\Psi: [\tau\str]\in H^1(G;\bbZ_d),[\tau]\in H_1(G;\bbZ_d)\right\}.
\end{equation*}
\end{thm}
\begin{proof}
(i) Immediate consequence of Lemma~\ref{lem: GS Classes}.

(ii) Theorem~\ref{thm:GS Alg} and Lemma~\ref{lem: GS Classes} imply that $\caA(\caS_G)$ is generated by $\left\{X_{\alpha\str,\beta\str} Z_{\alpha,\beta}: \alpha,\alpha\str,\beta,\beta\str\in(\bbZ_d)^g\right\}$. Furthermore,
\begin{equation}\label{Zab}
Z_{\alpha,\beta} = \prod_{i,j=1}^{g}Z_{v_i,0}^{\alpha_i} Z_{0,v_j}^{\beta_j},
\end{equation}
which proves that $\{X_{v_i,0}, X_{0,v_i}, Z_{v_i,0}, Z_{0,v_i}: 1\leq i\leq g\}$ is indeed a generating set. It remains to prove the commutation relations.  They follow from the fact that $\lambda_i$ and $\tau_i\str$ can be chosen to share only one spin, and no spin at all with other generators. Hence, the commutation relation of $Z_{v_i,0}$ and $X_{0,v_i}$ reduce to that of~(\ref{XZ}).

(iii) Let $\tilde Z$ be any of $Z_{v_j,0}$ or $Z_{0,v_j}$, and $\tilde X$ be such that $[[\tilde Z,\tilde X]] = \omega$. Since $\tilde Z$ is unitary and $\tilde Z^d = \mathbb I$, its spectrum is a subgroup of the $d$th roots of unity. Let $\psi\in\caS_G$ be such that $\tilde Z\psi = \mu\psi$, with $\mu^d = 1$. Then, $\langle\psi,  \tilde X \psi\rangle = \langle \tilde Z \psi,  \tilde X \tilde Z \psi\rangle = \omega^{-1}\langle\psi,  \tilde X \psi\rangle$ so that $\tilde X\psi\perp\psi$. In fact $\tilde X^j\psi$ is an eigenvector of $\tilde Z$ for the eigenvalue $\omega^j \mu$. Hence, $\mathrm{Sp}(\tilde Z) = \{\omega^j:0\leq j\leq d-1\}$. The algebra generated by $\tilde X,\tilde Z$ is isomorphic to $\caM_d(\bbC)$ generated by the original $X,Z$~(\ref{XZ}). Since it commutes with all other elements of $\caA(\caS_G)$,
\begin{equation*}
\caA(\caS_G)\cong \otimes_{i=1}^{b_1(G;\bbZ_d)}\caM_d(\bbC),
\end{equation*}
which yields the dimension of the ground state space.

(iv) follows from the isomorphism above and the fact that any vector is cyclic in $\bbC^d$ for $\caM_d(\bbC)$.
\end{proof}
We note that the two representations can also be viewed as a single representation of the homology group of $G\times G\str$. Indeed, K\"unneth's theorem states in the present case that
\begin{align*}
H_1(G\times G\str;\bbZ_d) &\cong \left(H_1(G;\bbZ_d)\otimes H_0(G\str;\bbZ_d)\right) \oplus\left( H_1(G\str;\bbZ_d)\otimes H_0(G;\bbZ_d)\right) \\
&\cong H_1(G;\bbZ_d) \oplus H^1(G;\bbZ_d)
\end{align*}
since $H_0(G;\bbZ_d) \cong \bbZ_d \cong H_0(G\str;\bbZ_d)$ because $G,G\str$ are connected.

\section{The entanglement entropy}\label{GSEE}

With the explicit description of Theorem~\ref{thm:SBraid} and its proof, it is immediate to provide a basis of the ground state space, which in turn allows for a simple expression for the entanglement entropy of any finite subset $\Lambda\subset\caG_1$ of edges. The loop representation developed here, which is by itself not new, provides a simple insight on the origin of both the area law and the topological term. The same conclusion could also be obtained from the `Projected Entangled Pair States' representation~\cite{Verstraete:2006jk} of the ground state, see~\cite{Schuch:2010jp}.

Let $\Omega = \otimes_{e\in\caG_1} l_0\in\caH_G$. Let $K$ be the group generated by $\{a_v:v\in\caG_0\}$ under multiplication, which is finite since $(a_v)^d = \bbI$ and the number of vertices is finite. We note that there is a one-to-one correspondence between $K$ and the $1$-coboundaries $\caB^1(G;\bbZ_d)$.
\begin{prop}\label{ExplicitGS}
Let
\begin{equation*}
\Psi_0 := \vert K \vert^{-1/2}\sum_{h\in K} h\Omega.
\end{equation*}
Then 
\begin{equation}\label{GS Basis}
\left\{\Psi_{\alpha,\beta}:=\prod_{i=1}^g X_{\alpha,\beta}\Psi_0: 0\leq \alpha_i,\beta_i\leq d-1\right\}
\end{equation}
is an orthonormal basis of $\caS_G$.
\end{prop}
\begin{proof}
As in the proof of Lemma~\ref{lem:AB Rep}, we have $[a_v,b_f] = 0$ for all vertices and faces. It follows that $ b_f\Psi_0 = \vert K \vert^{-1/2} \sum_{h\in K} h b_f \Omega = \Psi_0$ since $b_f\Omega = \Omega$ for all $f\in\caG_2$, and therefore $B_f\Psi_0 = \Psi_0$ by~(\ref{vertex and face operators}). Furthermore, $a_v\Psi_0 = \vert K \vert^{-1/2}\sum_{h\in K} a_vh \Omega = \vert K \vert^{-1/2}\sum_{g\in K} g \Omega = \Psi_0$ since $a_v$ is invertible, $a_v^{-1} = (a_v)^{d-1}$. Hence $A_v\Psi_0 = \Psi_0$, so that $\Psi_0\in\caS_G$. The remaining statements of the Proposition follow from the proof of Theorem~\ref{thm:SBraid}.
\end{proof}

Since $Z_{\alpha,\beta}$ is a cycle, then $[Z_{\alpha,\beta}, a_v]=0$ again by~(\ref{ZaCommutation}), so that $[Z_{\alpha,\beta}, h]=0$ for all $h\in K$. Hence $Z_{\alpha,\beta}\Psi_0 = \Psi_0$ as, moreover, $Z_{\alpha,\beta}\Omega = \Omega$. It follows from this and the argument in the proof of Theorem~\ref{thm:SBraid} that the basis~(\ref{GS Basis}) is a joint eigenbasis of all $Z_{\alpha,\beta}$ operators, and that the $X_{\alpha,\beta}$ operators are cyclic shifts within the basis.

The concrete representation of the ground states given in Proposition~\ref{ExplicitGS} allows for a simple computation of the entanglement entropy. Key here is the fact that $\Psi_0$ is an equal weight superposition of the action of all group elements on a product state, an observation originally made in~\cite{Hamma:2005je}.
\begin{thm}\label{thm:EE}
Let $\Lambda\subset \caG_1$. Let $\rho_{\alpha,\beta}(\Lambda)$ be the reduced density matrix of the state $\Psi_{\alpha,\beta}$, and let
\begin{equation*}
S_{\alpha,\beta}(\Lambda):= -\Tr\left(\rho_{\alpha,\beta}(\Lambda) \log \rho_{\alpha,\beta}(\Lambda)\right).
\end{equation*}
Then $S_{\alpha,\beta}(\Lambda) = S(\Lambda)$ for all $\alpha,\beta\in(\bbZ_d)^g$, with
\begin{equation}\label{Entropy}
S(\Lambda) = -\log \frac{\vert K_\Lambda \vert \vert K_{\Lambda^c}\vert}{\vert K \vert},
\end{equation}
where $\Lambda^c = \caG_1\setminus\Lambda$, and for any $Y\subset \caG_1$, $K_Y$ denotes the group generated by $\{a_v:e\ni v\Rightarrow e \in Y\}$.
\end{thm}
\begin{proof}
The fact that all entropies are equal follows from Proposition~\ref{ExplicitGS} and the observation that the $X_{\alpha,\beta}$ operators are products of one-site unitaries, so that all $\rho_{\alpha,\beta}({\Lambda})$ are unitarily equivalent. We therefore compute $S_{0,0}({\Lambda})$. Since any $a_v$ is a tensor product, any group element $h\in K$ can be decomposed into $h = t_{\Lambda}^h\otimes t_{\Lambda^c}^h$, where $t_{\Lambda}^h\in\caA_{\Lambda}$, but not necessarily in $K_{\Lambda}$. We first claim that
\begin{equation*}
t_{\Lambda}^h = t_{\Lambda}^{\tilde h} \Longleftrightarrow \tilde h = s h,\quad s\in K_{\Lambda^c}.
\end{equation*}
The implication $\Leftarrow$ being immediate, we consider $h,\tilde h$ such that $t_{\Lambda}^h = t_{\Lambda}^{\tilde h}$. That implies that $\tilde h = (1\otimes s_{\Lambda^c}) h$, namely $(1\otimes s_{\Lambda^c}) = \tilde h h^{-1}\in K$. But an element in $s\in K$ which is a simple tensor product of the form $s = (1\otimes s_{\Lambda^c})$ belongs to $K_{\Lambda^c}$. Indeed, $s$ can be represented as a polynomial in $\{X_{\gamma\str}:\caC^1(G;\bbZ_d)\}$ and $s\notin K_{\Lambda^c}$ implies that there is a cocycle having spins in both $\Lambda$ and $\Lambda^c$ so that $s$ cannot act trivially on $\Lambda$.

Let now $P_\Omega$ be the orthogonal projection onto $\Omega = \Omega_{\Lambda}\otimes\Omega_{\Lambda^c}$. We have
\begin{equation*}
\rho_{00}({\Lambda}) 
=\vert K\vert^{-1} \sum_{h,\tilde h\in K} \Tr_{\caH_{\Lambda^c}}\left(h P_\Omega \tilde h\str\right)
= \vert K\vert^{-1} \sum_{h,\tilde h\in K} t_{\Lambda}^h P_{\Omega_{\Lambda}} (t_{{\Lambda}}^{\tilde h})\str
\left\langle \Omega_{\Lambda^c}, (t_{\Lambda^c}^{\tilde h})\str t_{\Lambda^c}^h\Omega_{\Lambda^c}\right\rangle,
\end{equation*}
and the scalar product is non-zero if and only if $t_{\Lambda^c}^{\tilde h} = t_{\Lambda^c}^h$. By the initial claim,
\begin{equation*}
\rho_{00}({\Lambda}) = \vert K\vert^{-1} \sum_{h\in K, s\in K_{\Lambda}} t_{\Lambda}^h P_{\Omega_{\Lambda}} (t_{{\Lambda}}^{h})\str  s\str
\end{equation*}
Using the same remark, this can be further simplified as
\begin{equation*}
\rho_{00}({\Lambda}) = \frac{\vert K_{\Lambda^c}\vert }{\vert K\vert} \sum_{\substack{h\in K/K_{\Lambda^c} \\  s\in K_{\Lambda}}} t_{\Lambda}^h P_{\Omega_{\Lambda}} (t_{{\Lambda}}^{h})\str  s\str.
\end{equation*}
The theorem will follow from the identity
\begin{equation*}
\rho_{00}({\Lambda})\left(\rho_{00}({\Lambda}) - \frac{\vert K_{\Lambda} \vert \vert K_{\Lambda^c}\vert}{\vert K \vert} \bbI \right) = 0,
\end{equation*}
which proves that all non-zero eigenvalues of $\rho_{00}({\Lambda})$ are equal to $\frac{\vert K_{\Lambda} \vert \vert K_{\Lambda^c}\vert}{\vert K \vert}$, and the fact that $\Tr \rho_{00}({\Lambda}) = 1$. We have
\begin{align*}
\rho_{00}({\Lambda})^2 &= \frac{\vert K_{\Lambda^c}\vert^2 }{\vert K\vert^2}\sum_{\substack{h_1,h_2 \in K/K_{\Lambda^c} \\ s_2,s_2 \in K_{\Lambda}} } 
t_{\Lambda}^{h_1} P_{\Omega_{\Lambda}} (t_{{\Lambda}}^{h_2})\str s_2\str
\left\langle \Omega_{\Lambda^c}, (t_{\Lambda^c}^{h_1})\str s_1\str t_{\Lambda^c}^{h_2}\Omega_{\Lambda^c}\right\rangle \\
&=\frac{\vert K_{\Lambda^c}\vert^2 }{\vert K\vert^2}\sum_{\substack{h \in K/K_{\Lambda^c} \\ s_2,s_2 \in K_{\Lambda}}}   t_{\Lambda}^{h} P_{\Omega_{\Lambda}} (t_{{\Lambda}}^{h})\str s_1\str s_2\str.
\end{align*}
By the group property, the double sum over $s_1,s_2$ reduces to $\vert K_{\Lambda}\vert$ times a simple sum over $\{s:s\in K_{\Lambda}\}$, thereby yielding the claim.
\end{proof}
We define the interior $\Lambda_0$ and the boundary $\partial \Lambda$ of any $\Lambda\subset\caG_1$ as the following subsets of vertices:
\begin{align*}
\Lambda_0 &:= \{v\in\caG_0: \text{if } e\in\caG_1\text{ such that } e\ni v\text{, then } e \in \Lambda\}; \\
\partial \Lambda &:=\{v\in\caG_0: v\cap\Lambda\neq\emptyset\text{ and }v\cap\Lambda^c\neq\emptyset\}.
\end{align*}
Furthermore, we shall say that $\Lambda\subset \caG_1$ is \emph{simple} if $K_\Lambda$ is the free abelian group (over $\bbZ_d$) generated by $\{a_v:v\subset \Lambda_0\}$. Typical examples to keep in mind are simply connected subsets that do not `wind around the surface' at all. With this, we can prove the remarkable dependence of the entanglement entropy on the properly defined surface area of the set $\Lambda$, with a topological term $S_{\mathrm{topo}} = -\log d$ which is independent of the genus $g$.
\begin{cor}
Let $\Lambda\subset \caG_1$ be simple. Then,
\begin{equation}\label{EE}
S(\Lambda) = \left(\vert \partial \Lambda\vert - 1\right)\log d.
\end{equation}
\end{cor}
\begin{proof}
If $\Lambda$ is simple, then $\vert K_{\Lambda}\vert = d^{\vert {\Lambda_0}\vert}$ because the map
\begin{align*}
F: \quad &K_{\Lambda} \longrightarrow \bbZ_d^{\phantom{d}\Lambda_0} \\
&h= \prod_{v\in {\Lambda_0}} a_v^{j(v)} \longmapsto F(h) = \{F(h)_v = j(v) :v\in {\Lambda_0}\}
\end{align*}
is a bijection. However, since the full graph is embedded in a closed surface, the $\{a_v:v\in\caG_0\}$ are subject to a constraint, namely $\prod_{v\in\caG_0} a_v = \bbI$. Indeed, $\caG_1$ is not simple and in that case the map $F$ presented above is $1$-to-$d$, as $F(h)_v = j(v) + k$ represent the same group element for any $0\leq k\leq d-1$. Therefore, $\vert K\vert = d^{\vert \caG_0\vert-1}$. Noting that $\partial \Lambda = \caG_0 \setminus (\Lambda_0 \cup (\Lambda^c)_0)$ with $\Lambda_0 \cap (\Lambda^c)_0 = \emptyset$, we obtain from Theorem~\ref{thm:EE} that
\begin{equation*}
S(\Lambda) = -(\log d)\left(\vert \Lambda_0\vert + \vert (\Lambda^c)_0\vert - (\vert \caG_0 \vert -1) \right) = \left(\vert \partial \Lambda\vert - 1\right)\log d,
\end{equation*}
which is what we had set to prove.
\end{proof}

Equation~(\ref{EE}) can also be understood in terms of loop configurations. Indeed, $K$, identified with the set of $1$-coboundaries, is the set of all loop configurations with $d-1$ types of loops upon the empty state $\Omega$. A loop of type $j$ is created by $\prod a_v^j$, where the product is over all vertices, namely dual faces, the boundary of which make up the loop.

In this picture, $\frac{\vert K_\Lambda \vert \vert K_{\Lambda^c}\vert}{\vert K \vert}$ is the ratio of the number of loop configurations with no loop crossing the boundary of $\Lambda$ to the total number of loop configurations. An arbitrary loop configuration can be obtained from a `disjoint' one by the application of a number of $a_v$'s  with $v\in\partial\Lambda$. However, not all such operations create a crossing: $\prod_{v\in\partial \Lambda}a_v$ in fact creates two loops, one inside $\Lambda$, the other outside of it, with no crossing at all, see Figure~\ref{Fig:Loops}. On each of the free $\vert\partial \Lambda \vert-1$ boundary vertices $v$, the operator $a_v$ may act between $0$ and $d-1$ times, so that $\vert K \vert =  d^{\vert\partial \Lambda \vert-1} \vert K_\Lambda \vert \vert K_{\Lambda^c}\vert$ yielding~(\ref{EE}) again. This of course is nothing else than a restatement of the constraint used in the proof above since $\prod_{v\in\partial \Lambda}a_v = \prod_{v\in\Lambda_0}a_v^{-1}\prod_{v\in(\Lambda^c)_0}a_v^{-1}$, and because the outside loop is indeed a boundary in the toric topology. 
\begin{figure}
\includegraphics[width=0.6\textwidth]{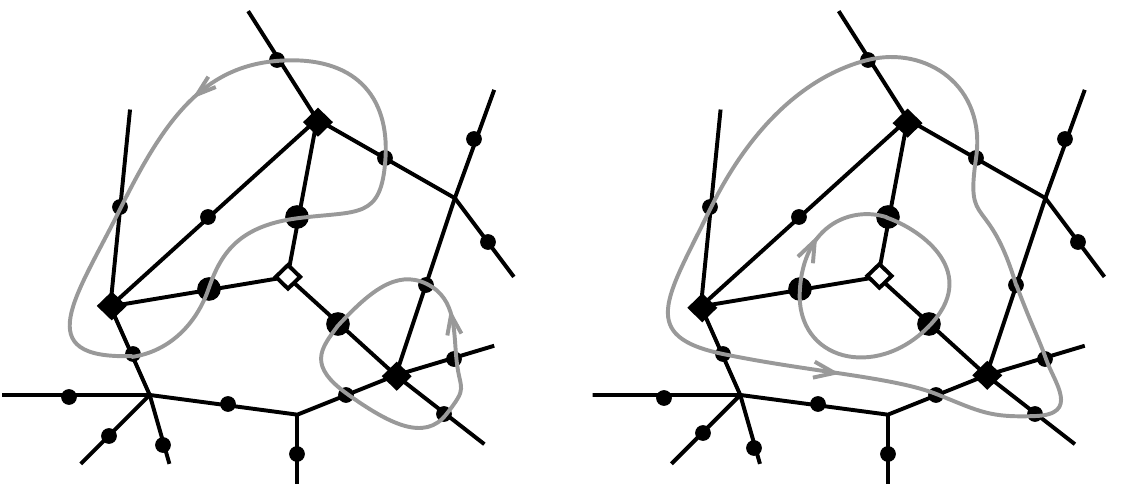}
\caption{In this figure, $\Lambda$ is the set of the three fattened circles, $\partial\Lambda$ the three black diamonds and $\Lambda_0$ the white diamond. The two figures depict the fact that acting on all vertices of $\partial\Lambda$ does create only apparent crossings. On the left is the direct loop interpretation, while the same action on the spins is interpreted on the right with two different loops, one completely inside $\Lambda$ and one completely outside of it.}
\label{Fig:Loops}
\end{figure}

Finally, we note that the topological term agrees, as is well-known, with the Kitaev-Preskill formula~\cite{Kitaev:2006dn}: As we shall see in the next section, there are $d^2$ abelian anyons in the theory, so that the total quantum dimension of the model is $D = \left(\sum_{\rho=1}^{d^2} 1^2\right)^{1/2} = d$. See also~\cite{Fiedler:2016up} for a very recent discussion of the quantum dimension and its relation to the topological entanglement entropy in the setting of algebraic quantum field theory.

\section{Excited states: Abelian anyons}\label{sec:Excited states}

We finally turn to the relation between the topological ground state degeneracy and the anyonic nature of the elementary excitations of the model. They are in this case equivalent characterisations of topological order. For notational simplicity, we restrict our attention to the torus, namely $g=1$. First of all, the excited state spaces are obtained by a unitary mapping from the ground state space $\caS_G$. Secondly, and in fact more importantly, the statistical phase of the distinguishable particles $X^k$ and $Z^l$, obtained by locally braiding one particle around the other, is equal to the only non-trivial combination of the global loop operators in $\caA(\caS_G)$.

As Section~\ref{Sec:GSS} has exposed, the degeneracy of the ground state space arises from the fact that it carries representations of the homology and cohomology groups of the cell decomposition. More generally, the cycle and cocycle groups are symmetry groups of the Hamiltonian.
\begin{prop}
For any $\gamma\in\caL_1(G;\bbZ_d)$, $[Z_\gamma,H_G] = 0$, and for any $\gamma\str\in\caL^1(G;\bbZ_d)$, $[X_{\gamma\str},H_G] = 0$.
\end{prop}
\begin{proof}
As noted in the proof of Theorem~\ref{thm:GS Alg}, $[Z_\gamma,A_v] = 0 = [Z_\gamma, B_f]$ for all $v\in\caG_0,f\in\caG_2$, whenever $\gamma\in \caL_1(G;\bbZ_d)$, and similarly for $X_{\gamma\str}$ for $\gamma\str\in\caL^1(G;\bbZ_d)$.
\end{proof}
\noindent In particular, each eigenspace of the Hamiltonian carries projective representations of both groups. 

A convenient way to describe the elementary particles is using walks. A \emph{walk} is an alternating sequence of vertices and edges starting and ending at a vertex, and it is naturally identified with a $1$-chain: for the walk $(v_0,e_1,v_1,e_2,\ldots e_N,v_N)$, its corresponding chain is $\gamma = \sum_{i=1}^N \varepsilon(e_i) e_i$, where $\varepsilon(e_i) = 1$ is $e_i$ is oriented from $v_i$ to $v_{i+1}$ and $\varepsilon(e_i) = -1$ if the edge has the opposite orientation. With this, $\partial\gamma = e_N - e_0$. Similarly, a \emph{cowalk} is an alternating sequence of faces and dual edges starting and ending at a face, and it is naturally identified with a $1$-cochain.

Let $\Psi\in\caS_G$ be an arbitrary ground state vector, and let
\begin{equation*}
\zeta^k_{\gamma} := Z^k_\gamma\Psi_{},\qquad \xi_{\gamma\str}^k := X_{\gamma\str}^k\Psi_{}
\end{equation*}
where $\gamma$ is a walk from $v$ to $\tilde v$, with $v\neq \tilde v$ and $1\leq k\leq d-1$, and similarly for the cowalk $\gamma\str$. $Z_\gamma^k$ commutes with all projectors in the Hamiltonian but two, namely $A_{v},A_{\tilde v}$. For those, we have
\begin{equation*}
A_v Z_\gamma^k\Psi_{} = Z_\gamma^k \frac{1}{d}\sum_{j=0}^{d-1}\omega^{-\varepsilon(e_1) (j+k)}(a_v)^j \Psi_{} = (\omega^{-\varepsilon(e_1)}Z_\gamma)^k \frac{1}{d}\sum_{j=0}^{d-1}\omega^{-\varepsilon(e_1) j} \Psi_{} = 0,
\end{equation*}
and similarly for the vertex $\tilde v$. Hence, since all $B_f$ terms commute with $Z_\gamma^k$,
\begin{equation*}
H_G\zeta_{\gamma}^k = (E_G + 2)\zeta_{\gamma}^k
\end{equation*}
where $E_G = -\left(\vert \caG_0\vert + \vert \caG_2\vert\right)$ is the ground state energy. We shall say that $\zeta_{\gamma}^k$ is a state having two elementary excitations located at $v$ and $\tilde v$. In this case, $a_v\zeta_{\gamma}^k = \omega^{\varepsilon(e_1) k} \zeta_{\gamma}^k$ while $a_{\tilde v}\zeta_{\gamma}^k = \omega^{-\varepsilon(e_N)  k} \zeta_{\gamma}^k$: the string operator $Z_\gamma^k$ creates a particle-antiparticle pair with `charges' $\pm k$ at its endpoints.

Note that although it may be tempting to describe $\zeta_{\gamma}^k$ as having $k$ identical particles, the fact that the energy is independent of $k$ shows that this is physically not a correct interpretation. There are indeed $2(d-1)$ elementary particles corresponding to the states $\{\zeta_\gamma^k,\xi_{\gamma\str}^j: 1\leq j,k\leq d-1\}$, and they can be fused to yield a total of $d^2$ particle states, including the vacuum, given by $\{X_{\gamma\str}^l Z_\gamma^k\Psi_{}: 0\leq k,l\leq d-1\}$.

By the group property, these particles are \emph{transportable} (in discrete steps, from any edge, resp.~face, to any other edge, resp.~face): If $\gamma_1$ is a walk from $v_0$ to $v_1$ and $\gamma_2$ a walk from $v_1$ to $v_2$, then $\gamma_1 + \gamma_2$ is a walk from $v_0$ to $v_2$, and $Z_{\gamma_2}$ transports the particle at $v_1$ to $v_2$. Once again, similar conclusions hold for $\xi_{\gamma\str}^l$. Therefore, there is a physically meaningful implementation of braiding operations. For simplicity, we consider here the case $k = l = 1$, the generalisation to general $k,l$ is immediate. Let $\gamma,\gamma\str$ be (co)walks from $v$ to $\tilde v$, respectively $f$ to $\tilde f$, with $v\neq \tilde v, f\neq \tilde f$. Then there is a closed walk $\tau$ from $v$ to itself, that is a boundary and such that $\tau\cap\gamma\str$ contains only one edge. The action of $Z_\tau$ on $\zeta_\gamma$ corresponds to transporting the $Z$-particle at $v$ around the loop $\tau$. In the additional presence of an $X$ particle, we therefore call the map
\begin{equation*}
\varphi = X_{\gamma\str}Z_\gamma\Psi \longmapsto B\varphi = Z_\tau X_{\gamma\str} Z_\gamma\Psi
\end{equation*}
a \emph{physical braiding} of the $Z$ particle around the $X$ particle.

The following theorem characterises the set of excited states describing two particle-antiparticle pairs placed at fixed positions on the lattice. We freely use the identification of $v\str\in\caG^0$ with $f\in\caG_2$.
\begin{thm}\label{thm:ExcitedBraids}
Let $v,\tilde v\in\caG_0$ and $f,\tilde f\in\caG_2$, with $v\neq \tilde v, f\neq \tilde f$, and let $\Psi\in\caS_G$. Let
\begin{equation*}
\caV_{XZ}:= \mathrm{span}\left\{X_{\gamma\str}Z_\gamma\Psi: \gamma\text{ is a walk from }v\text{ to }\tilde v,\,\gamma\str\text{ is a cowalk from }f\text{ to }\tilde f \right\}.
\end{equation*}
Then
\begin{enumerate}
\item $\caV_{XZ} = X_{\gamma\str} Z_\gamma\caS_G$ for any fixed $\gamma\str ,\gamma$; \\ 
in particular,
\begin{enumerate}
\item $\mathrm{dim}(\caV_{XZ}) = \mathrm{dim}(\caS_G)$
\item $\caV_{XZ}$ does not depend on $\Psi$,
\end{enumerate}
\item for any $\varphi\in\caV_{XZ}$, 
\begin{equation}\label{Equivalence}
B\varphi = \left(X_{1,0} Z_{0,1} X_{1,0}^{-1} Z_{0,1}^{-1}\right) \varphi
\end{equation}
and $B\varphi = \omega^{-1}\varphi$.
\end{enumerate}
\end{thm}
In other words, the anyonic property of the particles $X$ and $Z$, namely the local map $B$ on the l.h.s~of~(\ref{Equivalence}), is equivalent to the non-trivial structure of the ground state space, which is completely characterised by the combination $X_{1,0} Z_{0,1} X_{1,0}^{-1} Z_{0,1}^{-1}$ of non-local operators corresponding to non-trivial cycles on the r.h.s.~of~(\ref{Equivalence}).

\begin{proof}
(i) If $\gamma,\tilde\gamma$ are two walks with the same initial and final vertices, then $\gamma - \tilde\gamma$ is a cycle, and $Z_\gamma\Psi_{} = Z_{\tilde\gamma} Z_{\gamma - \tilde\gamma}\Psi_{}$ with $Z_{\gamma - \tilde\gamma} \Psi\in\caS_G$ and $Z_{\gamma - \tilde\gamma} \Psi $ depends only on the homology class of $\gamma - \tilde\gamma$ by Lemma~\ref{lem: GS Classes}. Similarly, let $\gamma\str,\tilde\gamma\str$ be two cowalks from $f$ to $\tilde f$. Then again, $X_{\gamma\str} Z_\gamma\Psi = X_{\tilde \gamma\str}X_{\gamma\str-\tilde\gamma\str}  Z_\gamma\Psi = \omega^k X_{\tilde \gamma\str}Z_\gamma X_{\gamma\str-\tilde\gamma\str}  \Psi$ which depends only on $[\gamma\str-\tilde\gamma\str]$,  and $0\leq k\leq d-1$ depends on the number of oriented crossings of the closed cowalk $\gamma\str-\tilde\gamma\str$ with the walk $\gamma$. It follows that
\begin{equation*}
\caV_{XZ}= \mathrm{span}\left\{X_{\gamma\str}Z_\gamma X_{\tau\str}Z_{\tau}\Psi: [\tau\str]\in H^1(G;\bbZ_d),[\tau]\in H_1(G;\bbZ_d)\right\},
\end{equation*}
where $\gamma\str ,\gamma$ are arbitrary but fixed (co)walks, and further $\caV_{XZ}=X_{\gamma\str} Z_\gamma\caS_G$ by Theorem~\ref{thm:SBraid}. This and the unitarity of $X_{\gamma\str} Z_\gamma$ yields immediately the equality of dimensions, while the cyclicity of any $\Psi\in\caS_G$ implies that $\caV_{XZ}$ is independent of $\Psi$.

\begin{figure}
\includegraphics[width=\textwidth]
{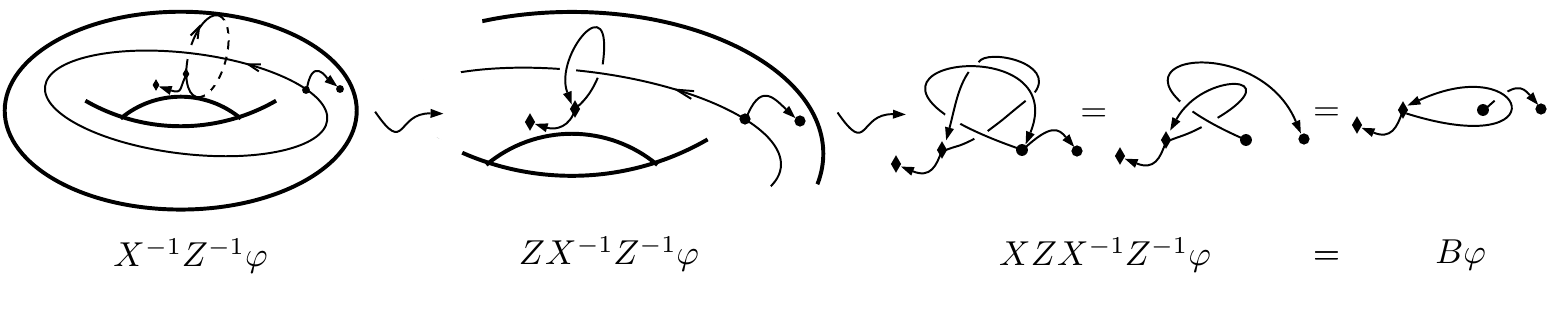}
\caption{A diagrammatic proof of~(\ref{Equivalence}). An operator acting on the left of another one is represented by a line drawn above the other line. The strings can be deformed by adding $1$-boundaries.}
\label{Fig:Fundamental}
\end{figure}
Finally, the proof of~(\ref{Equivalence}) follows from simple algebraic manipulations using only (i) and the fact that operators corresponding to cycles depend only on the (co)homology class when acting on the ground state space, Lemma~\ref{lem: GS Classes}. This can best be represented graphically as in~Figure~\ref{Fig:Fundamental}, where the ordering of operators is given by the superposition of strings, while the possible addition of (co)boundaries corresponds to the deformation of the strings. With this, the fact that $B\varphi = \omega^{-1}\varphi$ follows from Theorem~\ref{thm:SBraid}(ii).
\end{proof}

We note that $\caV_{XZ}$ does only carry a \emph{projective} unitary representation of $H_1(G;\bbZ_d)$. Indeed, a closed walk $\tau$ that is a boundary has $Z_\tau = \omega^{\frac{1}{2}(n(\tau)-\tilde n(\tau))} \bbI$ on $\caV_{XZ}$, where $0\leq n(\tau),\tilde n(\tau)\leq d-1$ is the winding number of $\tau$ around the face $f$, resp. $\tilde f$. Hence, $\tau$ detects the total charge of $X$ particles in its interior. Similarly, a closed cowalk detects the total charge of $Z$ particles in its interior.

So far, we have discussed the mutual statistics of $X$ and $Z$ as distinguishable particles. Besides, particles of the same type are indistiguishable among each other, and their properties under permutations can easily be obtained. First of all, if many $Z^k$ particles are present, they are bosons since their associated string operators commute, and similarly for $X^l$ particles. Now, for $Z^kX^l$ particles with both $k,l\neq 0$, the permutation of a pair of such particles yields a phase $\omega^{-(k+l)}$, which can be schematically represented (in the case $k=l=1$) as in Figure~\ref{Fig:Braids}.
\begin{figure}
\includegraphics[width=0.75\textwidth]
{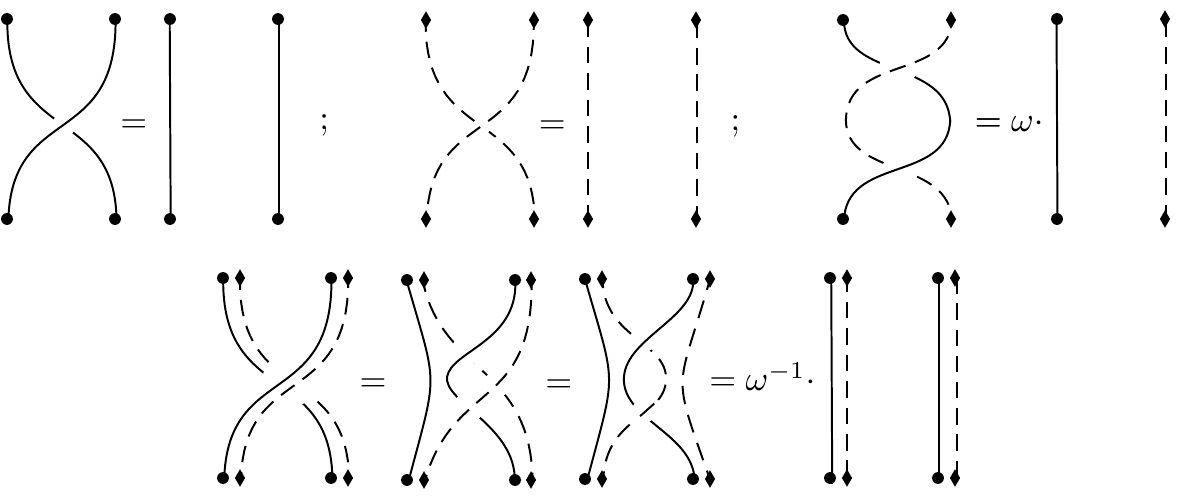}
\caption{The elementary braids of Kitaev's abelian models. $Z$ and $X$ particles are bosons (first two diagrams), their mutual statistics is given by a phase $\omega$ (third diagram); Together, these facts imply that the $ZX$ particle is an anyon with statistical phase $\omega^{-1}$ (second line): The first, resp.~second, equality follows from the bosonic nature of $Z$, resp. $X$, and the last one from the mutual statistics.}
\label{Fig:Braids}
\end{figure}

Summarising, the elementary excitations of Kitaev's model for a fixed polygon decomposition $G$ of $T_g$ can be referred to as \emph{discrete anyons}. The distinguishable particles $X^l$ and $Z^k$ have a non trivial mutual statistics arising from the representation of $H_1(G\times G\str;\bbZ^d)$ on the ground state space, while they are bosons among themselves. The composite particles, as indistiguishable particles, are abelian anyons. Here, the concept of anyon is associated with homology and cohomology groups of the discrete cell decomposition of the manifold $T_g$. These groups are closely related -- but not equal -- to the geometric braid group over $T_g$, which is, as the fundamental group of the suitable configuration space, related to homotopy, see~\cite{Birman:1969aa}. Not surprisingly, a suitably reinterpreted~(\ref{Equivalence}) also holds for the braid group (Corollary~5.1 in~\cite{Birman:1969aa}) as pointed out in~\cite{Einarsson:1990vf} and applied is the present context in~\cite{Kitaev:2003ul}. 

\section{Conclusion}

We have studied Kitaev's abelian models defined on finite polygon decompositions of closed orientable surfaces. We have in particular exhibited the equivalence of two characterisations of topological order, namely the topological ground state degeneracy and the anyonic nature of elementary particles. Both are closely related to the fact that the ground state space carries a unitary representation of the first homology group of the product of the cell complex with its dual. The concrete equivalence can be stated as an equality between a local cycle acting on the excited state space and a non-trivial combination of the global cycle and cocycle acting on the ground state space, both yielding the behaviour of the distinguishable particles under braiding. Not surprisingly, this also corresponds to a property of the two-strings unpermuted geometric braid group on the torus.

These features can be expected to be found in generic topologically ordered systems: Degenerate ground states related to each other by `Wilson loops', excited states created from the ground state space by string localised observables, anyonic statistics arising from commutation relations of the loop operators, and topological entanglement entropy. Broadly speaking, this is in fact the framework put forward in~\cite{Levin:2005fy}, and developed in~\cite{Chen:2010gb} to complement Landau's classification of phases based on local order parameters.

\subsection*{Acknowledgement.} I wish to thank H.~Schulz-Baldes for a careful reading of this article and numerous comments and suggestions.

\end{document}